%% file: main.tex
\documentclass{article}

\title{Quantum State Synthesis: Relation with Decision Complexity Classes and Impossibility of Synthesis Error Reduction}
\date{June 2024}

\input{preambule}

\usepackage{framed}

\begin{document}

\author[1,2]{Hugo Delavenne}
\author[3]{François Le Gall}

\affil[1]{Institut Polytechnique de Paris, Ecole Polytechnique, LIX, Palaiseau, France}
\affil[2]{INRIA}
\affil[3]{Graduate School of Mathematics, Nagoya University}
%\affil[4]{ENS Paris-Saclay, Université Paris-Saclay}

\date{}
\maketitle

\begin{abstract}
    This work investigates the relationships between quantum state synthesis complexity classes (a recent concept in computational complexity that focuses on the complexity of preparing quantum states) and traditional decision complexity classes. We especially investigate the role of the \emph{synthesis error parameter}, which characterizes the quality of the synthesis in quantum state synthesis complexity classes. %We introduce a new definition of quantum state synthesis complexity classes, which slight generalizes definitions from prior works, 
    We first show that in the high synthesis error regime, collapse of synthesis classes implies collapse of the equivalent decision classes. 
    For more reasonable synthesis error, we then show a similar relationships for $\BQP$ and $\QCMA$. Finally, we show that for quantum state synthesis classes it is in general impossible to improve the quality of the synthesis: unlike the completeness and soundness parameters (which can be improved via repetition), the synthesis error cannot be reduced, even with arbitrary computational power. 
\end{abstract}

\section{Introduction}

While quantum complexity theory traditionally investigates the complexity of decision problems (i.e., Boolean functions), a recent line of research \cite{Aar16,BEMPQY23,DLLM23,Ros24,MY23,RY22} started investigating the complexity of constructing quantum states, a task called \emph{quantum state synthesis}. Those prior works showed that the behavior of quantum state synthesis complexity classes is often similar to the behavior of decision complexity classes: the equality $\PSPACE=\QIP$ \cite{JJUW11} has its state synthesis equivalent $\sPSPACE=\sQIP$ \cite{RY22,MY23}, the equality $\QCMA=\QCMA[1,\frac12]$ \cite{JKNN12} has the equivalent $\sQCMA=\sQCMA[1,\frac12]$ \cite{DLLM23}, and the equality $\QIP=\QIP(O(1))$ \cite{Wat03} has the equivalent $\sQIP=\sQIP(O(1))$ \cite{Ros24}.

In order to further investigate the relationships between quantum state synthesis complexity classes and decision complexity classes, in this paper we introduce a new definition of quantum state synthesis complexity classes in which we allow an arbitrary number of target states per input (prior definitions required exactly one target state per input). Our definition is closer to the definition of classical functional classes like \FP{} or \FNP{} introduced in \cite{MP89} and thus more closely related to Boolean classes and languages. We stress that this new definition is a generalization of the definitions of prior works: we recover the previous definitions as a special case by requiring one target state per input.

Based on this new definition, we further explore the relationship between quantum state synthesis complexity classes and decision complexity classes. Here are our main contributions:
\begin{itemize}
\item[(1)]
We investigate the relationship between the class $\BQP$ and the corresponding quantum state synthesis complexity class denoted $\rsBQP_\delta$. Here $\delta\in[0,1]$ is a parameter called the \emph{synthesis error parameter} that characterizes the imprecision of the synthesis (the goal is to have $\delta$ as small as possible). \Cref{prop:equivalence-state-decision} shows that if we take $\delta$ very close to 1, i.e., if we allow exponentially small fidelity between the ouput and the state we want to synthesize, there exists a tight relationship between $\BQP$ and $\rsBQP_\delta$. This relationship remains true for other complexity classes (e.g., $\QMA$ and $\rsQMA_\delta$, or $\QCMA$ and $\rsQCMA_\delta$).
\item[(2)] 
This above result yields the question of proving relationships between quantum state synthesis complexity classes and decision complexity classes for more reasonable values of the parameter $\delta$. We make a first step in this direction. We especially investigate how proving separations for quantum state synthesis classes relates to proving separations for decision complexity classes. We first observe that  $\BQP\neq \QMA$ implies $\rsQMA_\delta\not\subseteq\rsBQP_\delta$ for all $\delta$ (\Cref{cor:no-collapse})  and $\BQP\neq \QCMA$ implies $\rsQCMA_\delta\not\subseteq\rsBQP_\delta$ for all~$\delta$ (\Cref{cor:no-collapse2}).
Our main contribution (in \Cref{prop:guess-classical-witness}) proves the converse for the case of $\QCMA$ (with a small loss in~$\delta$): if there exist $\delta$ and a polynomial $q$ such that $\rsQCMA_\delta\not\subseteq\rsBQP_{\delta+1/q}$, then $\BQP\neq\QCMA$. These results suggest that progress on understanding decision complexity classes can be done by investigating quantum state synthesis classes.
\item[(3)]
We finally investigate whether the synthesis error parameter $\delta$ can be reduced, i.e., whether the quality of the synthesis can be increased, just like completeness and soundness can be improved via repetition. We show that for quantum state synthesis classes this is in general impossible: we prove (see \Cref{cor:impossibility}) that $\rsBQP_\delta\not\subset\rsBQP_{\delta-\epsilon}$ holds for any $\epsilon>0$. We actually prove in \Cref{thm:impossibility} that reducing $\delta$ is impossible even if we allow arbitrary computational power. This result holds for the definitions used in prior works as well and shows the importance of the parameter~$\delta$ when defining state synthesis complexity classes. This result is closely related to the impossibility of error reduction for unitary synthesis problems of \cite[Proposition 3.8]{BEMPQY23}, but it is stronger in the sense that there is no gap between the source and target errors.
\end{itemize}

\section{Definition of relational state synthesis complexity classes}

\label{sec:definitions}
We first recall the definition of classical functional classes \cite{MP89}.
In this work, we always  use the binary alphabet $\Sigma=\{0,1\}$.

\begin{defn}[\FP, \FNP, \TFNP]
    A relation $R\subseteq\Sigma^*\times\Sigma^*$ is in $\FP$ iff there exists a polynomial-time Turing machine $M$ such that if there exists $y\in\Sigma^*$ such that $(x,y)\in R$ then $M(x)$ outputs such a $y$, and otherwise $M(x)$ rejects.

    A relation $R\subseteq\Sigma^*\times\Sigma^*$ is in $\FNP$ iff there exists a polynomial-time Turing machine $M$ such that if there exists $y\in\Sigma^*$ such that $(x,y)\in R$ then there exists $w\in\Sigma^*$ such that $M(x,w)$ outputs such a $y$, and otherwise, for any $w\in\Sigma^*$, $M(x,w)$ rejects.

    A relation $R\in\FNP$ is in $\TFNP$ iff $\forall x\in\Sigma^*,\exists y\in\Sigma^*, (x,y)\in R$.
\end{defn}

The relations for state synthesis are a bit more complex since we have to specify the output space for every input size.

\begin{defn}[State synthesis relation]
    For $n\in\bbN$, let $\calH_n$ be a Hilbert space and $\mathcal{O}_n$ be the set of density matrices over $\calH_n$.
    A \emph{state synthesis relation} is a triple $(R,\Lyes,\Lno)$ where $(\Lyes,\Lno)$ is a promise language and
    \[
    R=\{(x,\rho)\mid x\in \Lyes, \rho\in S_x\}%\subseteq\Lyes \times \mathcal{O}
    \]
    for non-empty subsets $S_x\subseteq \mathcal{O}_{|x|}$.
    We often omit the language. We simply use $R$ to denote the state synthesis relation, we write $\Lyes_R=\Lyes$, $\Lno_R=\Lno$ and $L_R=\Lyes\cup \Lno$ and define 
 \[
 xR:=\{\rho\in \mathcal{O}_{|x|}\mid (x,\rho)\in R\}
 \]
 for any $x\in L_R$ (note that $xR\notin\emptyset$ for any $x\in\Lyes_R$ and $xR=\emptyset$ for any $x\in\Lno_R$). We also define a function $k_R:\bbN\to\bbN$ that gives the number of qubits of $\calH_n$.
\end{defn}

%An important property is required for a relation $R$: for every $x\in L_R$, $xR$ must be convex.

The quantum circuits considered in this paper are bounded in size and uniform. We give formal definition of these notions.

\begin{defn}[Polynomial-size family of circuits]
    A family of quantum circuits $(C_n)_{n\in\bbN}$ is said to be polynomial-size if there exists a polynomial $p$ such that for any $n\in\bbN$, $C_n$ contains at most $p(n)$ gates.
\end{defn}

\begin{defn}[Uniform family of circuits]
    A family of quantum circuits $(C_n)_{n\in\bbN}$ is said to be polynomial-time-uniform, or simply uniform, if there exists a Turing machine $M$ working in polynomial-time such that for any $n\in\bbN$, $M(n)$ outputs a description of $C_n$.
\end{defn}

Due to the continuity of the space of quantum states, we need a measure and a threshold to quantify the tolerated error on the state synthesis. We use the trace distance between density matrices, and extend it to a distance between a density matrix and a set of density matrices.

\begin{defn}[Trace distance]
    Let $\rho$ and $\sigma$ be two density matrices on the same space. Define 
    \[
    \td(\rho,\sigma):=\frac12\Tr\left(\sqrt{(\rho-\sigma)^\dagger(\rho-\sigma)}\right).
    \]
    For a density matrix $\rho$ and a set $S$ of density matrices over the same space, define 
    \[
    \td(\rho,S):=\min_{\sigma\in S}\td(\rho,\sigma).
    \]
\end{defn}

%If a set $S$ is empty, then $\td(\rho,S)=1$.

%In this work, when there is a family of circuits $C_n$ taking some classical input $\ket{x}$ and possibly some other input $\ket{\psi}$ (for a witness), we denote by $C_x(\ket{\psi})$ the circuit $C_{|x|}(\ket{x}\ket{\psi})$. When $C_x(\ket{\psi})$ outputs an acceptance bit, we say that it is a decision circuit and we denote by $C_x^\acc(\ket{\psi})$ this acceptance bit (after measurement). When $C_x(\ket{\psi})$ has also an output channel, we denote by $C_x^\out(\ket{\psi})$ the quantum state outputted on this channel (as in \cite{DLLM23}, we allow to trace out some qubits in the resulting state); if $C$ has an acceptance bit we denote $C_x^\outacc(\ket{\psi})$ the quantum state outputted on this channel when $C_x^\acc(\ket{\psi})=1$.

In this work we consider families of circuits $C_n$ taking some classical input $\ket{x}$ and possibly some other input $\ket{\psi}$ (for a witness). We denote by $C_x(\ket{\psi})$ the circuit $C_{|x|}(\ket{x}\ket{\psi})$.
The circuit has a specific qubit that is measured at the end of the computation. The measurement outcome is called the acceptance bit and denoted by $C_x^\acc(\ket{\psi})$. 
%The circuit may additionally have an output channel. If it has one, we say that the circuit is a synthesis circuit; otherwise we say that the circuit is a decision circuit. 
When there is an output channel, we denote by $C_x^\out(\ket{\psi})$ the quantum state outputted on this channel. Note that this state depends on the value of the acceptance bit. We denote by $C_x^\outacc(\ket{\psi})$ the quantum state outputted on this channel when $C_x^\acc(\ket{\psi})=1$. When there is no witness we remove $\ket{\psi}$ from all these notations, e.g., we use write the acceptance bit simply as $C_x^\acc$.

We are now ready to introduce relational state synthesis complexity classes.

\begin{defn}[\rsBQP]
    \label{def-sBQP}
    Let $c,s,\delta\colon\bbN\to[0,1]$ be completeness, soundness and synthesis error functions.
    A state synthesis relation $R$ is in $\rsBQP_\delta[c,s]$ if there exists a uniform family of polynomial-size quantum circuits $(C_n)_{n\in\bbN}$ such that for $x\in L_R$:
    \begin{itemize}
    \item completeness: if $xR\neq\emptyset$ then $\Pr\big(C_{x}^\acc=1\big)\geq c(|x|)$ and $\td(C_{x}^\outacc,xR)\leq\delta(k_R(|x|))$.
    \item soundness: if $xR=\emptyset$ then $\Pr\big(C_{x}^\acc=1 \big)\leq s(|x|)$.
    
    %\fnote{Nouvelle tentative pour eviter les probabilites conditionelles: if $xR\neq\emptyset$ then $\Pr\big(C_{x}^\acc=1\text{ and }C_x^\out=\rho \big)\leq s(|x|)$ for any state $\rho$ such that $\td(\rho,xR)>\delta(k_R(|x|))$; and if $xR=\emptyset$, $\Pr\big(C_{x}^\acc=1 \big)\leq s(|x|)$. Qu'en penses-tu?}

    %\hnote{Mon but était d'introduire des probabilités conditionnelles pour être plus convaincu des preuves qu'avec la définition des papiers précédents qui était ``if $\td(xR,C^\out)>\delta(k_R(|x|))$ then $\Pr(C^\acc=1)\leq s(|x|)$''. Pour éviter par exemple de faire des contraposées sans invoquer le théorème de Bayes. Je suis d'accord que ce que vous écrivez est plus intuitif que la définition qu'on utilise mais je ne suis pas sûr que c'est équivalent aux définitions des papiers précédents (et aussi à ce qu'on utilise dans les preuves ici). On peut écrire votre définition $\Pr\big(C_{x}^\acc=1 \mid\text{ and } \td(C_{x}^\out,xR)>\delta(k_R(|x|))\big)\leq s(|x|)$, et donc il y a un facteur $\Pr\big(\td(C_{x}^\out,xR)>\delta(k_R(|x|))\big)$ entre votre proposition et la définition originelle. (D'ailleurs je crois que ce facteur m'a souvent embêté, donc j'aurais sûrement préféré travailler avec votre définition.)}
    \end{itemize}
\end{defn}

\begin{defn}[\rsQMA]
    \label{def-sQMA}
    Let $c,s,\delta\colon\bbN\to[0,1]$ be functions.
    A state synthesis relation $R$ is in $\rsQMA_\delta[c,s]$ if there exists a uniform family of polynomial-size quantum circuits $(C_n)_{n\in\bbN}$ such that for $x\in L_R$:
    \begin{itemize}
    \item completeness: if $xR\neq\emptyset$ then there exists a quantum witness $\ket{\psi}$ such that $\Pr\big(C_{x}^\acc(\ket{\psi})=1\big)\geq c(|x|)$.
    \item soundness: for any $\ket{\psi}$, if both $xR\neq\emptyset$ and $\td(C_{x}^\outacc(\ket{\psi}),xR)>\delta(k_R(|x|))$ hold then $\Pr\big(C_{x}^\acc(\ket{\psi})=1\big)\leq s(|x|)$; and if $xR=\emptyset$ then $\Pr\big(C_{x}^\acc(\ket{\psi})=1\big)\leq s(|x|)$.
    \end{itemize}
\end{defn}

The definitions used in the previous papers \cite{RY22,MY23,Ros24,BEMPQY23,DLLM23} do not involve relations since exactly one output is expected per input. We rephrase the definition of \s\BQP{} and \sQMA{} from \cite{DLLM23} by using our definitions of \rsBQP{} and \rsQMA.

\begin{defn}[\s\BQP, \s\QMA]
    \label{def-fake-sQMA}
   For any $c,s,\delta\colon\bbN\to[0,1]$,
    \begin{align*}
    \s\BQP_\delta[c]&=\{R\in\rsBQP_\delta[c,0]\mid \forall x\in L_R, |xR|=1\},\\
    \s\QMA_\delta[c,s]&=\{R\in\rsQMA_\delta[c,s]\mid \forall x\in L_R, |xR|=1\}.
    \end{align*}
\end{defn}

We define the class \rsQCMA{} similarly to \rsQMA{} but with a restriction to witnesses being states in the computational basis (i.e., classical strings).
%The following classes are defined similarly to \rsQMA{}: 
%or \rsBQP:
%\begin{itemize}
    %\item \rsQCMA{} with a restriction to witnesses being states in the computational basis (i.e. classical strings),
    %\item \rsQIP{} with an interaction with a quantum prover instead of being given a single state.
    %\item \rsPSPACE{} with a polynomial-space-uniform family of polynomial-space circuits instead of a polynomial-time-uniform polynomial-size family.
%\end{itemize}
We also define a class \rsR{} corresponding to states synthesized by arbitrary (uniform) quantum circuits (this class can be seen as the equivalent of the class $\mathsf{R}$ of recursive languages in decision complexity theory):

\begin{defn}[\rsR]
    \label{def-sR}
    Let $c,s,\delta\colon\bbN\to[0,1]$ be completeness, soundness and synthesis error functions.
    A state synthesis relation $R$ is in $\rsR_\delta[c,s]$ if there exists an unboundedly powerful Turing machine such that for any $n\in\bbN$, $M(1^n)$ halts and outputs the description of a quantum circuit $C_n$ such that for $x\in L_R$:
    \begin{itemize}
    \item completeness: if $xR\neq\emptyset$ then $\Pr\big(C_{x}^\acc=1\big)\geq c(|x|)$ and $\td(C_{x}^\outacc,xR)\leq\delta(k_R(|x|))$.
    \item soundness: if $xR=\emptyset$ then $\Pr\big(C_{x}^\acc=1\big)\leq s(|x|)$.
    \end{itemize}
\end{defn}

% A very general definition is given in \Cref{sec:general-definition}, along with a reformulation of some results generalized for every classes.

Finally, we show that the gap between the completeness and the soundness can be amplified for these classes. For \rsBQP{}, \rsQCMA{} and \rsQMA{}, the proof of gap amplification of \cite{DLLM23} applies directly since it amplifies the completeness and soundness while preserving the target state:
\begin{prop}[Gap amplification]
    \label{amplification}
    Let $0\leq c(n),s(n),\delta(n)\leq 1$ be poly-time computable functions such that $c(n)-s(n)\geq1/\poly(n)$. For any polynomial $p$,
    \begin{align*}
        \rsBQP_\delta[c,s]&\subseteq\rsBQP_\delta[1-2^{-p},2^{-p}] \\
        \rsQCMA_\delta[c,s]&\subseteq\rsQCMA_\delta[1-2^{-p},2^{-p}]\\
        \rsQMA_\delta[c,s]&\subseteq\rsQMA_\delta[1-2^{-p},2^{-p}]\text.
    \end{align*}
\end{prop}
Since we have an amplification for completeness and soundness, having completeness $2/3$ and soundness $1/3$ is equivalent to having completeness $c(n)$ and soundness $s(n)$ with $c(n)-s(n)\geq1/\poly(n)$. We thus define  
\begin{align*}
    \rsBQP_\delta&=\rsBQP_\delta[2/3,1/3],\\
    \rsQCMA_\delta&=\rsQCMA_\delta[2/3,1/3],\\
    \rsQMA_\delta&=\rsQMA_\delta[2/3,1/3].
\end{align*}

%simply refer by $\rsBQP_\delta$ to $\rsBQP_\delta[2/3,1/3]$, by $\rsQCMA_\delta$ to $\rsQCMA_\delta[2/3,1/3]$ and by $\rsQMA_\delta$ to $\rsQMA_\delta[2/3,1/3]$.
A gap amplification result for the class $\rsR$ is also easy to show:
\begin{prop}[Gap amplification]
    \label{amplification-R}
    Let $0\leq c(n),s(n),\delta(n)\leq 1$ be computable functions such that $c(n)>s(n)$. For any computable function $\gamma(n)>0$,
    \[\rsR_\delta[c,s]\subseteq\rsR_\delta[1-\gamma,\gamma]\text.\]
\end{prop}
\begin{proof}
    The amplification is very similar to the standard amplification for decision circuits by repetition: We repeatedly apply the synthesis circuit until we get $C_{x}^\acc=1$. As soon at this happens, we stop and output the output state of the last repetition.  If we do not get $C_{x}^\acc=1$ after a specified number of interactions (depending on $c$, $s$ and $\gamma$), we decide that $xR=\emptyset$ (note that there is no need to output a quantum state in this case).
\end{proof}

%========================================================================
\section{High synthesis error regime}\label{sec:high}
%========================================================================

State synthesis classes defined in Section \ref{sec:definitions} are closely related to decision languages. In \Cref{prop:equivalence-state-decision} below we show a basic relationship between these two notions when the synthesis error is close to 1. While for concreteness we focus on the relationship between the classes $\BQP$ and $\rsBQP_\delta$, the results proved in this section remains true for other complexity classes (e.g., $\QMA$ and $\rsQMA_\delta$, $\QCMA$ and $\rsQCMA_\delta$, or $\QIP$ and $\rsQIP_\delta$) as well.

We  start with the following lemma, which holds for any $\delta$. 
\begin{lemma}
    \label{lem:equivalence-state-decision}
    For any $\delta:\bbN\to[0,1]$ and any state synthesis relation $R$, if $R\in\rsBQP_\delta$ then $L_R\in\BQP$.
\end{lemma}
\begin{proof}
    Take $R\in\rsBQP_{\delta}$. Let $(C_n)_{n\in\bbN}$ denote the family of circuits from \Cref{def-sBQP}. By ignoring the output state of the circuits and considering only their acceptance qubit, they become decision circuits that have acceptance probability $\Pr(C_x^\acc=1)$.
    If $x\in \Lyes_R$ then $\Pr(C_x^\acc=1)\geq 2/3$ by completeness as a state synthesis circuit.
    If $x\in\Lno_R$, which means that $xR=\emptyset$, then $\Pr(C_x^\acc=1)\leq 1/3$ by soundness as a state synthesis circuit. Thus $L_R\in\BQP$.
\end{proof}

Next, we show a tight relationship between quantum state synthesis classes and decision complexity classes when $\delta=1-2^{-n}$, i.e., when we allow exponentially small fidelity between the output and the state we want to synthesize. The idea is to generate the same maximally mixed state on any input.
\begin{thm}%[Equivalence between state synthesis and decision classes with exponentially bad error]
    \label{prop:equivalence-state-decision}
    Consider the function $\delta_0\colon n\mapsto1-2^{-n}$. Then for any state synthesis relation $R$, $R\in\rsBQP_{\delta_0}$ iff $L_R\in\BQP$.
\end{thm}

\begin{proof}
    From \Cref{lem:equivalence-state-decision} we immediately get that $R\in\rsBQP_{\delta_0}$ implies $L_R\in\BQP$.

    Now suppose that $L_R\in\BQP$ and let $(C_n)_{n\in\bbN}$ be a uniform family of circuits recognizing $L_R$ with completeness $2/3$ and soundness $1/3$. A maximally mixed state $\rho_n$ on $k_R(n)$ qubits can be synthesized by a uniform family of polynomial-size circuits because $k_R\in\poly$. Since $\rho_n$ is at distance at most $1-2^{-k_R(n)}$ from any other density matrix, the circuit $C'_n$ that outputs $C_n^{\prime\acc}=C_n^\acc$ and $C_n^{\prime\out}=\rho_n$ synthesizes $R$ with completeness $2/3$, soundness $1/3$ and error $\delta_0$. Thus $R\in\rsBQP_{\delta_0}$.
\end{proof}

%========================================================================
\section{Relationship between decision and state synthesis classes}
%========================================================================
In this section we investigate the relationship between proving separations for quantum state synthesis classes and proving separations for decision complexity classes. The results of this section hold for any value of the synthesis error parameter $\delta$. Our main result is \Cref{prop:guess-classical-witness}, which shows that a separation for quantum state synthesis classes can be used to prove a separation for decision complexity classes.

First, as a consequence of \Cref{lem:equivalence-state-decision}, we show the following result:

\begin{prop}
    \label{cor:no-collapse}
    %If there exists $\delta,\delta':\bbN\to[0,1]$ such that $\rsQMA_\delta\subset\rsBQP_{\delta'}$ then $\BQP=\QMA$.

    If $\BQP\neq\QMA$ then 
    \[
    \rsQMA_\delta\not\subseteq\rsBQP_{\delta'}
    \]
    holds for any $\delta,\delta'\in[0,1]$.
\end{prop}

\begin{proof}
    Suppose that there exist $\delta,\delta'$ such that the inclusion $\rsQMA_\delta\subseteq\rsBQP_{\delta'}$ holds.
    For $L=(\Lyes,\Lno)\in\QMA$, there exists a family of circuits $(C_n)_{n\in\bbN}$ that takes a quantum witness and recognizes $L$ with completeness $2/3$ and soundness $1/3$. For each $n\in\bbN$ we can build a circuit $C'_n$ that introduces a 1-qubit output channel and ``copies'' the contents of the acceptance qubit to the output channel using a CNOT gate:

    
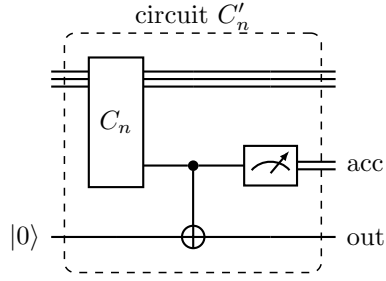
\begin{center}
    \begin{quantikz}
        & \gate[2,nwires={2},bundle={1}]{C_n}\gategroup[wires=3,steps=3,style={inner sep=2mm, dashed, line width=0.2mm, rounded corners}]{circuit $C'_n$} & \qwbundle[alternate]{} & \qwbundle[alternate]{} & \qwbundle[alternate]{} & &  \\
        &  & \ctrl{1} & \meter{} & \cw{}\rstick{$\acc$} & & \\
        \lstick{\ket{0}} & \qw & \targ{} & \qw & \qw\rstick{$\out$} &&
    \end{quantikz}
    \end{center}

    Define the relation $R$ by $xR=\{\ketbra{1}{1}\}$ if $x\in\Lyes$ and $xR=\emptyset$ if $x\in\Lno$. Since $(C'_n)_{n\in\bbN}$ synthesizes $R$, we get
    \[
    R\in\rsQMA_0\subseteq\rsQMA_\delta\subseteq\rsBQP_{\delta'}.
    \]
    By \Cref{lem:equivalence-state-decision}, we get $L\in\BQP$ and thus $\QMA\subseteq\BQP$.
\end{proof}

By replacing the quantum witness by a classical witness in \Cref{cor:no-collapse} we similarly obtain the following result:  
\begin{prop}
    \label{cor:no-collapse2}
    %If there exists $\delta,\delta':\bbN\to[0,1]$ such that $\rsQMA_\delta\subset\rsBQP_{\delta'}$ then $\BQP=\QMA$.

    If $\BQP\neq\QCMA$ then 
    \[
    \rsQCMA_\delta\not\subseteq\rsBQP_{\delta'}.
    \]
    holds for any $\delta,\delta'\in[0,1]$.
\end{prop}

Now using a technique similar to the proof that $\P=\NP$ iff $\FP=\FNP$ \cite{BG94} we are able to show the following converse statement, which is the main result of this section.
\begin{thm}
    \label{prop:guess-classical-witness}
    If there exist some $\delta\in[0,1]$ and some polynomial $q$ such that
    \[
    \rsQCMA_\delta\not\subseteq\rsBQP_{\delta+1/q},
    \]
  then $\BQP\neq\QCMA$.
    %Then $\BQP=\QCMA$ iff $\forall\delta,\rsBQP_\delta[c,s]=\rsQCMA_\delta[c,s]$.
    %Then $\BQP\neq \QCMA$ if and only if there $\exists$ $\delta\in[0,1]$ such that $\rsQCMA_\delta[c,s]\not\subseteq\rsBQP_\delta[c,s]$.
\end{thm}

The proof of \Cref{{prop:guess-classical-witness}} will show the contrapositive: we will show that $\BQP=\QCMA$ implies that $\rsQCMA_\delta\subseteq\rsBQP_{\delta+1/q}$ holds for any $\delta\in[0,1]$ and any polynomial $q$.
In order to prove this statement, we first show (in \Cref{GW2} below) that if $\BQP=\QCMA$ then we can efficiently ``guess'' the witness of the circuit synthesizing a relation in $\rsQCMA$. For conciseness, we will write 
\[
    f_p(\ell,n)=1-2^{-n}-\frac{\ell}{p(n)^2}.
\]
for a polynomial $p:\bbN\to\bbN$ and any integers $\ell,n$.

\begin{prop}[Guessing a classical witness]
    \label{GW2}
    Let $R$ be a relation in the complexity class $\rsQCMA_\delta[1-2^{-n},2^{-n}]$ for some $\delta>0$, and $(C_n)_{n\in\bbN}$ be the corresponding family of quantum circuits synthesizing $R$. Let $p$ be a polynomial such that the circuit $C_n$ acts on less than $p(n)$ qubits, and $\ell(n)$ be the length of the classical witness it receives. If $\BQP=\QCMA$, then there exists a polynomial-time quantum algorithm that receives as input a string $x\in\Lyes_R$ and outputs with probability at least $(1-2^{-n})^{\ell(n)}$ a string $w\in\{0,1\}^{\ell(n)}$ such that 
    \[\Pr(C_x^\acc(\ket{w})=1)\geq f_p(\ell(n)+1,n)\] holds.
\end{prop}

We use the following lemma to prove \Cref{GW2}.

\begin{lemma}
    \label{lemma:GW}
    For any polynomial $p:\bbN\to\bbN$, the promise language $\GW_{p}:=(\GWyes_p,\GWno_p)$ defined below is in $\QCMA$.
    \begin{align*}
        \GWyes_{p} &:= \left\{(C,x,w_0)\mid
            \begin{cases}
                C\text{ describes a quantum circuit taking }\leq p(|x|)\text{ qubits as input}\\
                |x|+|w_0|<p(|x|)\\
                \exists w, \Pr(C^\acc(\ket{x}\ket{w_01w})=1)
                %=1
                %\geq 1-2^{-|x|}
                \geq f_p(|w_0|,|x|)
            \end{cases}\right\} \\
        \GWno_{p} &:= \left\{(C,x,w_0)\mid
            \begin{cases}
                C\text{ describes a quantum circuit taking }\leq p(|x|)\text{ qubits as input}\\
                |x|+|w_0|<p(|x|)\\
                \forall w, \Pr(C^\acc(\ket{x}\ket{w_01w})=1)\leq 
                %2^{-|x|}\text.
                f_p(|w_0|+1,|x|)\text.
            \end{cases}\right\}
    \end{align*}
\end{lemma}

\begin{proof}
    For any $n\in\bbN$, consider the following verification circuit. The circuit receives as input $(C,x,w_0)$ and $w$ as classical witness.
    It simulates $C(\ket{x}\ket{w_01w})$ and accepts iff this simulation accepts.\vspace{2mm}

    \noindent\textbf{Completeness.} If $(C,x,w_0)\in\GWyes_{p}$ then there exists $w$ such that 
    \[
        \Pr(C^\acc(\ket{x}\ket{w_01w})=1)\geq f_p(|w_0|,|x|)
    \]
    holds.\vspace{2mm}

    \noindent\textbf{Soundness.} If $(C,x,w_0)\in\GWno_{p}$, then for any $w$, the inequality
    \[
        \Pr(C^\acc(\ket{x}\ket{w_01w})=1)\leq f_p(|w_0|+1,|x|)
    \]
    holds. \vspace{2mm}

    Since $f_p(|w_0|,|x|)-f_p(|w_0|+1,|x|)$ is lower bounded by an inverse-polynomial function of the input length, we conclude that $\GW_p\in\QCMA$.
   % Therefore $\GW_p\in\QCMA[c,s]$ where $c-s$ is inverse-polynomial, which implies $\GW_p\in\QCMA$.
\end{proof}

We are now ready to give the proof of \Cref{GW2}.
\begin{proof}[Proof of \Cref{GW2}]
    Assume that $\BQP=\QCMA$. Let $\mathcal{A}$ be a polynomial-time quantum algorithm deciding $\GW_{p}\in\BQP[1-2^{-n},2^{-n}]$, where $\GW_{p}$ is defined in \Cref{lemma:GW}. Consider the following quantum algorithm that receives $x\in\Lyes_R$ as input. The algorithm constructs bit by bit a classical witness $w=w_1...w_{\ell(n)}$ by defining the bit $w_i$ as follows: if $\mathcal{A}$ on input $(C_n,x,w_1...w_{i-1})$ accepts then set $w_i=1$, otherwise set $w_i=0$. 

    This running time of this algorithm is polynomial. We now show its correctness. Consider a string $x\in \Lyes_R$. In the analysis below, we assume that $V$ does not make any error (i.e., always decides correctly membership in $\GW_{p}$ during the $\ell(n)$ iterations), which happens with probability at least $(1-2^{-n})^{\ell(n)}$.  
    
    For conciseness, for any $q\in[0,1]$ we say that a string $\bar{w}\in\{0,1\}^{\ell(n)}$ is a $q$-witness if 
    \[
        \Pr(C_x^\acc(\ket{\bar{w}}))\geq q
    \]
    holds. For conciseness again, we write below $f(i)$ instead of $f_p(i,n)$.
    
    We show by induction on $i$ that for each $i\in\{0,\ldots,\ell(n)\}$ the following property $\mathcal{P}_i$ holds at the end of the $i$th iteration (or at the very beginning of the algorithm for $i=0$): there exists an $f(i+1)$-witness starting with $w_1\ldots w_i$.  Property $P_{\ell(n)}$ then implies the correctness of our algorithm.
    
    Property $P_0$ is obviously true: from the completeness of $C_n$ we know that there exists at least one $f(0)$-witness. 

%    Assume now that the property $P_{i-1}$ is true for some $i\in\{1,\ldots,\ell(n)\}$, i.e., there exists an $f(i)$-witness starting with $w_1\ldots w_{i-1}$. If there exists at least such an $f(i)$-witness in which the next bit is $1$ then $(C_n,x,w_1...w_{i-1})\in\GWyes_{p}$,  which means that $\mathcal{A}$ on input $(C_n,x,w_1...w_{i-1})$ accepts and we correctly set $w_i=1$. Otherwise there is at least one $f(i)$-witness in which the next bit is 0. If there is no $f(i+1)$-witness in which the next bit is 1 then $(C_n,x,w_1...w_{i-1})\in\GWno_{p}$, which means that Algorithm $\mathcal{A}$ rejects and we set $w_i=0$; otherwise the output of $\mathcal{A}$ (and the value of $w_{i+1}$) can be arbitrary, which is fine since in this case there exist both an $f(i+1)$-witness in which the next bit is 1 and an $f(i)$-witness in which the next bit is 0. Since $f(i+1)\geq f(i)$, in all cases Property $P_i$ is thus satisfied.

Assume now that the property $P_{i-1}$ is true for some $i\in\{1,\ldots,\ell(n)\}$, i.e., there exists an $f(i)$-witness starting with $w_1\ldots w_{i-1}$. If there exists an $f(i)$-witness starting with $w_1\ldots w_{i-1}1$ then $(C_n,x,w_1...w_{i-1})\in\GWyes_{p}$,  which means that $\mathcal{A}$ on input $(C_n,x,w_1...w_{i-1})$ accepts and we correctly set $w_i=1$. Otherwise there exists an $f(i)$-witness starting with $w_1\ldots w_{i-1}0$. If there is no $f(i+1)$-witness starting with $w_1\ldots w_{i-1}1$ then $(C_n,x,w_1...w_{i-1})\in\GWno_{p}$, which means that Algorithm $\mathcal{A}$ rejects and we correctly set $w_i=0$; otherwise the output of $\mathcal{A}$ (and the value of $w_{i}$) can be arbitrary, which is fine since in this case there exist both an $f(i+1)$-witness starting with $w_1\ldots w_{i-1}1$ and an $f(i)$-witness starting with $w_1\ldots w_{i-1}0$. Since $f(i+1)\geq f(i)$, an $f(i)$-witness is an $f(i+1)$-witness. In all cases Property $P_i$ is thus satisfied.
\end{proof}

We can now apply \Cref{GW2} to prove \Cref{prop:guess-classical-witness}.

\begin{proof}[Proof of \Cref{prop:guess-classical-witness}]
    We show the contrapositive: we show that $\BQP=\QCMA$ implies that for any $\delta\in[0,1]$ and any polynomial $q$, the class $\rsQCMA_\delta$ is included in $\rsBQP_{\delta+1/q}$.
    
    %Assume that $\BQP=\QCMA$ and take any relation $R\in\rsQCMA_\delta[c,s]$ with $c(n)=1-2^{-n}$ and $s(n)=2^{-n}$. Let $p$ be a polynomial such that $p(n)\geq\sqrt{2q(n)(\ell(n)+1)}$
    %and $(C_n)_{n\in\bbN}$ denote the circuit synthesizing $R$, with the same notations as in the statement of \Cref{GW2}, \nouveau{and assume that $C_n$ takes less than $p(n)$ qubits as input.} Let $C'_n$ be the circuit obtained by first applying the circuit corresponding to the algorithm of \Cref{GW2} to guess a witness $w$ and then simulating $C_n(\ket{x}\ket{w})$. In the following, let $X$ be the random variable that gives the witness $w$ guessed by $C'_x$, and for conciseness let $d:=\td(C_x^{\prime\outacc},xR)$, $d_w:=\td(C_x^{\outacc}(\ket{w}),xR)$ and $\delta:=\delta(k_R(n))$.\vspace{2mm}

    Assume that $\BQP=\QCMA$ and take any relation $R\in\rsQCMA_\delta[c,s]$ with $c(n)=1-2^{-n}$ and $s(n)=2^{-n}$.
    Let $(C_n)_{n\in\bbN}$ denote the circuit synthesizing $R$ with completeness $c$, soundness $s$ and synthesis error $\delta$, let $\ell(n)$ be the length of the classical witness $C_n$ receives and let $p_1(n)$ be the number of qubits that $C_n$ takes as input.
    Let $p$ be a polynomial such that $p(n)\geq\sqrt{2q(n)(\ell(n)+1)}$ and $p(n)\geq p_1(n)$ hold. Let $C'_n$ be the circuit obtained by first applying the circuit corresponding to the algorithm of \Cref{GW2} to guess a witness $w$ and then simulating $C_n(\ket{x}\ket{w})$. In the following, let $X$ be the random variable that gives the witness $w$ guessed by $C'_x$, and for conciseness let $d=\td(C_x^{\prime\outacc},xR)$, $d_w=\td(C_x^{\outacc}(\ket{w}),xR)$ and $\delta=\delta(k_R(n))$.\vspace{2mm}

    \noindent\textbf{Completeness.} Suppose that $xR\neq\emptyset$, i.e., $x\in\Lyes_R$. Then by \Cref{GW2} we obtain
    \begin{align*}
        \Pr(C_{x}^{\prime\acc}=1) &\geq (1-2^{-n})^{\ell(n)}\cdot  f_{p}(\ell(n)+1,n)
        \\&= \underset{\geq 1-\ell(n)2^{-n}}{\underbrace{(1-2^{-n})^{\ell(n)}}}\left(1-2^{-n}-\frac{\ell(n)+1}{p(n)^2}\right)
        \\&\geq 1-(\ell(n)+1)2^{-n}-\frac{\ell(n)+1}{p(n)^2}\\
        &\geq 1-2^{\log(\ell(n)+1)-n}-\frac{1}{2q(n)}
        =:c'(n)\text,
    \end{align*}
    where the last inequality holds since we chose a polynomial $p$ satisfying $p(n)\geq\sqrt{2q(n)(\ell(n)+1)}$.

    Denote $p_\delta=\Pr(d_w\leq\delta)=\sum_{d_w\leq\delta}\Pr(X=w)$. Since
    \begin{align*}
        c'(n)
          &\leq \Pr(C_x^{\prime\acc}=1)
        \\&=\sum_{d_w>\delta}\Pr(X=w)\underset{\leq s(n)}{\underbrace{\Pr(C_x^{\acc}(\ket{w})=1)}}+\sum_{d_w\leq\delta}\Pr(X=w)\underset{\leq1}{\underbrace{\Pr(C_x^{\acc}(\ket{w})=1)}}
        \\&\leq  s(n)+p_\delta\text,
    \end{align*}
    we have
    \begin{align*}
        d &\leq \sum_{w}\Pr(X=w)d_w
        \\&\leq \sum_{d_w\leq\delta}\Pr(X=w)\underset{\leq \delta}{\underbrace{d_w}}+\sum_{d_w>\delta}\Pr(X=w)\underset{\leq 1}{\underbrace{d_w}}
        \\&\leq \delta \underset{\leq 1}{\underbrace{p_\delta}} + 1-\underset{\geq c'(n)-s(n)}{\underbrace{p_\delta}}
        \\&\leq \delta+1-c'(n)+s(n)
        \\&\leq \delta+2^{\log(\ell(n)+1)-n}+\frac{1}{2q(n)}+2^{-n}
        \\&\leq \delta+\frac{1}{q(n)}
    \end{align*}
    when $2^{\log(\ell(n)+1)-n}+2^{-n}\leq \frac{1}{2q(n)}$, which holds when $n$ is large enough.

    \medskip\noindent\textbf{Soundness.} Suppose that $xR=\emptyset$, i.e., $x\in\Lno$. Then by soundness of $C_n$, whatever the witness $w$ guessed, the acceptance probability is small:
    \[\Pr\big(C_x^{\prime\acc}=1\big)= \sum_w\Pr(X=w)\underset{\leq s(n)}{\underbrace{\Pr(C_x^\acc(\ket{w})=1)}}\leq s(|x|)\text.\]

    \medskip{}Since $c'(n)-s(n)\geq 1/\poly(n)$, we obtain the inclusion $R\in\rsBQP_{\delta+1/q}$.
\end{proof}

\Cref{prop:guess-classical-witness} yields the question of achieving the same result with a quantum witness (i.e., the converse of \Cref{cor:no-collapse}).

\begin{conjecture}
    Does 
    $
    \rsQMA_{\delta}\not\subseteq\rsBQP_{\delta+1/p}
    $
    for some $\delta:\bbN\to[0,1]$ and polynomial $p:\bbN\to\bbN$ imply $\BQP\neq \QMA$? 
\end{conjecture}

The technique of guessing a quantum witness by using the assumption $\BQP=\QMA$ 
%and the fact that a relation $R$ is in $\rsQMA$ iff $L_R\in\QMA=\BQP$ 
could not be used here (except if $\QCMA=\QMA$) because using this technique would mean that there is a way to create a valid $\QMA$ witness by using classical information and with a polynomial-size circuit.

%=====================================================
\section{Impossibility to reduce the synthesis error}
%======================================================
In this section we prove that it is impossible to reduce the synthesis error for the class $\rsBQP$.
Here is the main result:

\begin{thm}
    \label{thm:impossibility}
    For any $0<\epsilon(n)\leq\delta(n)\leq1-2^{-n}$ and $0\leq s(n)<c(n)\leq1$, 
    \[\rsBQP_\delta[1,0]\not\subset\rsR_{\delta-\epsilon}[c,s].\]
\end{thm}

\Cref{thm:impossibility} shows the impossibility to reduce $\delta$ even when arbitrary computational power is available and even without a gap between $c$ and $s$. The following is a straightforward corollary:
\begin{corollary}
    \label{cor:impossibility}
    For any $0<\epsilon(n)\leq\delta(n)\leq1-2^{-n}$, \[\rsBQP_\delta\not\subset\rsBQP_{\delta-\epsilon}.\]
\end{corollary}

This result holds for any class, as long as it is possible to synthesize the maximally mixed state. It holds for the definitions used in prior works as well, even when considering only the inputs in unary \cite{RY22,MY23,DLLM23} as we actually do in the proof of \Cref{thm:impossibility}.

\begin{proof}[Proof of \Cref{thm:impossibility}]
    %We construct with a diagonal argument a family of strings that cannot be generated by a Turing machine with probability strictly greater than uniform, and then use it to construct quantum states that can be approximated by a mixed state with error $\delta$ but such that generating it with error $>\delta$ implies that the string can be generated with probability $>2^{-n}$.

    We use a diagonal argument to construct a family of strings that cannot be generated with non-trivial probability, and then use it to construct quantum states that can be approximated by a mixed state with error $\delta$ but such that generating it with error strictly smaller than~$\delta$ implies that the family of strings can be generated with non-trivial probability.

    \medskip\noindent\textbf{Constructing strings from a diagonal argument.} Since we are considering uniform families of quantum circuits (i.e., families of quantum circuits generated by Turing machines) we can enumerate them. Let $\mathcal{C}^1$, $\mathcal{C}^2$, $\ldots$ be such an enumeration and for each $r\in\bbN$, let $\mathcal{C}^r=\left\{ C_n^r\right\}_{n\in\bbN}$ denote the circuits in the family. 
    
    For any $r\in\bbN$, we focus on the circuit $C_r^r$, i.e., we take $n=r$ (as usual in diagonal arguments). Let $k(r)$ denote the number of qubits of the output channel of the circuit $C_r^r$. Consider the procedure of Figure \ref{fig1}, which we call Procedure $\mathcal{P}$.
    \begin{figure}
    \begin{framed}
        \noindent Procedure $\mathcal{P}$
    \begin{itemize}
        \item[1.]
        Apply the circuit $C^r_r$ on the initial state $\ket{0}^{\otimes r}$.
        \item[2.]
         Measure the qubit corresponding to the acceptance bit. Let $b\in\{0,1\}$ denote the outcome. 
        \item[3.]
         Measure the output channel in the computational basis. Let $z\in\{0,1\}^{k(r)}$ denote the outcome.
    \end{itemize}   
\end{framed}\vspace{-2mm}
\caption{Procedure $\mathcal{P}$}
\label{fig1}
\end{figure}
    
        For any string $z\in\{0,1\}^{k(r)}$, let $p(z)$ denote  the probability of obtaining $z$ at Step 3 conditioned on getting $b=1$ at Step 2.
         From a straightforward counting argument, there is at least one $z$ such that 
         \[
            p(z)\le 2^{-k(r)}.
        \]  
        We denote this string (or one of them, chosen arbitrarily, if there are more than one) by $u_r$.

    \medskip\noindent\textbf{Constructing the relation.} For each $n\in\bbN$, define the quantum state
    \[
        \rho_n^\delta=\sum_{z\in\{0,1\}^{k(n)}}\alpha_{z}\ketbra{z}{z},        
    \] 
    where
    \[\alpha_{z}=
        \begin{cases}
            2^{-{k(n)}}+\delta(n)&\text{ if }z=u_n, \\
            2^{-{k(n)}}-\frac{1}{2^{-{k(n)}}-1}\delta(n)&\text{ otherwise.}
        \end{cases}
    \]
    Define the relation $R^\delta=\{(0^n,\rho_n^\delta)\:|\:n\in\bbN\}$. 
    %by $R^\delta$ $xR^\delta=\{\rho_{|x|}^\delta\}$. 
    Note that the maximally mixed state on $k(n)$ qubits is at distance $\delta(n)$ from $\rho_n^\delta$. Since the maximally mixed state can be generated by a polynomial-size circuit with probability $1$, we get 
    \[
        R^\delta\in\rsBQP_\delta[1,0].
    \]

    \medskip\noindent\textbf{Impossibility to generate the relation with error \boldmath{$<\delta$}.} Suppose that there exists a uniform circuit family that synthesizes $R^\delta$ with error $\delta-\epsilon<\delta$ and completeness and soundness $c>s$. From \Cref{amplification-R} we can assume without loss of generality that $c(n)\ge 1-\gamma(n)$ for some (computable) function $\gamma$ such that 
    \[
    0<\gamma(n)<1- \frac{2^{-k(n)}}{2^{-k(n)} +\epsilon(n)}.
    \]
    %have a family of uniform quantum circuits such synthesizes $R^\delta$ with error $\delta-\epsilon$ and completeness greater than $q(n)$. 
    Let $\mathcal{C}^r=\left\{ C_n^r\right\}_{n\in\bbN}$ be this family, for some $r\in\bbN$. 
    %Consider the following process: apply the circuit $\mathcal{C}^r_r$ on the initial state $\ket{x}$ and then measure the output channel of the circuit in the computational basis.
    Apply Procedure $\mathcal{P}$ described above on the circuit~$C_r^r$. 
    
    Consider $p(u_r)$, the probability of obtaining the string $u_r$ at Step 3 of the procedure conditioned on getting $b=1$ at Step 2. Observe that measuring the state $\rho_r^\delta$ in the computational basis gives outcome $u_r$ with probability 
    \[
        2^{-k(r)} + \delta(r).
    \]
    
      By completeness, the probability that $C_r^r$ accepts and generates a state at distance at most $\delta(r)-\epsilon(r)$ from $\rho_r^\delta$ is greater than $1-\gamma(r)$. 
      %Note that
      % Since the circuit $C_r^r$ generates a state at distance at most $\delta(r)-\varepsilon(r)$ from $\rho_r^\delta$ (conditioned on $b=1$), we have
      %\begin{align*}
      %  2^{k(r)}q(r)\left(2^{-k(r)} +\epsilon(r)\right)
      %  & = \frac{1+2^{k(r)}\epsilon(r)}{1+2^{k(r)-1}\epsilon(r)} > 1
      %\end{align*}
      We thus have
      \begin{align*}
        p(u_r) &\ge 
        \left(1-\gamma(r)\right)\left(2^{-k(r)} +\delta(r) - \big(\delta(r)-\epsilon(r)\big)\right) \\
        &= \left(1-\gamma(r)\right)\left(2^{-k(r)} +\epsilon(r)\right) > 2^{-k(r)},
      \end{align*}
which is impossible by the construction of $u_r$.
\end{proof}

\section*{Acknowledgement}

Most of this work was done while HD was in visit at Nagoya University and financed by ENS Paris-Saclay, Université Paris-Saclay. The authors thank Yupan Liu and Masayuki Miyamoto for the fruitful discussions that happened there. FLG was supported by JSPS KAKENHI Grants Nos.\ JP20H00579, JP20H05966, JP20H04139, JP24H00071 and MEXT Quantum Leap Flagship Program (MEXT Q-LEAP) Grant No.~JPMXS0120319794.

\bibliographystyle{alpha}
\bibliography{biblio}

\end{document}

%% file: preambule.tex
\usepackage{amsfonts,amsmath,amssymb,amsthm}
\usepackage[left=3cm,right=3cm,top=3cm,bottom=3cm]{geometry}
\usepackage{graphicx}
\usepackage{tikz}
\usetikzlibrary{quantikz}
\usepackage{hyperref}
\usepackage[noabbrev,nameinlink]{cleveref}
\usepackage{authblk}

\newtheorem{thm}{Theorem}
\newtheorem{defn}{Definition}
\newtheorem{lemma}{Lemma}
\newtheorem{prop}{Proposition}
\newtheorem{corollary}{Corollary}
\newtheorem{conjecture}{Open question}

\newcommand{\poly}{\ensuremath{\mathrm{poly}}}

\newcommand{\td}{\ensuremath{\mathrm{td}}}
\newcommand{\ketbra}[2]{\ket{#1}\!\!\bra{#2}}

\newcommand{\Tr}{\ensuremath{\mathrm{Tr}}}

\newcommand{\cfont}[1]{\ensuremath{\mathsf{#1}}}
\newcommand{\pfont}[1]{\ensuremath{\mathbf{#1}}}

\newcommand{\acc}{\mathrm{acc}}
\newcommand{\out}{\mathrm{out}}
\newcommand{\outacc}{{\out|\acc}}

\newcommand{\NP}{\cfont{NP}}
\renewcommand{\P}{\cfont{P}}
\newcommand{\FNP}{\cfont{FNP}}
\newcommand{\TFNP}{\cfont{TFNP}}
\newcommand{\FP}{\cfont{FP}}

\newcommand{\s}{\cfont{state}}

\newcommand{\rs}{\cfont{relationalState}}

\newcommand{\BQP}{\cfont{BQP}}
\newcommand{\QMA}{\cfont{QMA}}
\newcommand{\QCMA}{\cfont{QCMA}}

\newcommand{\QIP}{\cfont{QIP}}

\newcommand{\PSPACE}{\cfont{PSPACE}}
\newcommand{\R}{\cfont{R}}

\newcommand{\sQMA}{\s\QMA}
\newcommand{\sQCMA}{\s\QCMA}
\newcommand{\sQIP}{\s\QIP}

\newcommand{\sPSPACE}{\s\PSPACE}

\newcommand{\rsBQP}{\rs\BQP}
\newcommand{\rsQMA}{\rs\QMA}
\newcommand{\rsQCMA}{\rs\QCMA}
\newcommand{\rsQIP}{\rs\QIP}

\newcommand{\rsR}{\rs\R}

\newcommand{\Lyes}{L^\mathrm{yes}}
\newcommand{\Lno}{L^\mathrm{no}}

\newcommand{\GW}{\pfont{GW}}
\newcommand{\GWyes}{\GW^\mathrm{yes}}
\newcommand{\GWno}{\GW^\mathrm{no}}

\newcommand{\sqcip}[1]{}
\newcommand{\constant}[1]{}
\newcommand{\tgw}[1]{}

\renewcommand{\epsilon}{\varepsilon}
\renewcommand{\emptyset}{\varnothing}
\renewcommand{\bar}[1]{\overline{#1}}

\newcommand{\calH}{\ensuremath{\mathcal{H}}}

\newcommand{\bbN}{\ensuremath{\mathbb{N}}}